\DocumentMetadata{testphase=new-or-1} 
\documentclass[11pt]{article}
\usepackage{amsmath}
\usepackage{amssymb}
\usepackage{amsthm}
\usepackage{graphicx} 
\usepackage[margin=1in]{geometry}
\usepackage{hyperref}
\usepackage{xcolor}
\usepackage{multicol}
\usepackage{xurl}
\usepackage[utf8]{inputenc}
\usepackage{enumitem}
\usepackage{fontawesome}
\newcommand{\gd}{\color{blue}{\faDatabase}}
\newcommand{\bd}{{\color{red}{\faTimesCircle}}}
\newcommand{\forkv}{{\rotatebox[origin=c]{-90}{$\in$}}} 
\usepackage{tikz}
\usetikzlibrary{shapes.geometric}
\usetikzlibrary{arrows.meta}
\usetikzlibrary{calc}
\def\arraystretch{1.25}

\newcommand{\zp}{\operatorname{zp}}
\newcommand{\ZP}{\operatorname{ZP}}
\newcommand{\row}{\operatorname{row}}
\newcommand{\col}{\operatorname{col}}

\newcommand{\rank}{\operatorname{rank}}

\newcommand{\wt}{\operatorname{wt}}

\newcommand{\F}{\mathbb{F}}
\newcommand{\B}{\mathbb{B}}

\newcommand{\eps}{\varepsilon}

\newcommand{\recov}{\operatorname{recov}}

\newcommand{\R}{\mathbb{R}}

\newcommand{\MDS}{\operatorname{MDS}}
\newcommand{\LDMDS}{\operatorname{LD-MDS}}

\newtheorem{theorem}{Theorem}

\newtheorem{proposition}[theorem]{Proposition}

\theoremstyle{definition}
\newtheorem{definition}[theorem]{Definition}

\title{\textbf{Maximal Recoverability: A Nexus of Coding Theory}\footnote{This article is an extended version of a survey appearing in IEEE BITS. If one plans to cite this survey, please cite the magazine version~\cite{brakensiek2025maximal}.}}
\author{Joshua Brakensiek\thanks{JB is with the Department of Electrical Engineering and Computer Sciences, University of California, Berkeley. Contact: \href{mailto:josh.brakensiek@berkeley.edu}{josh.brakensiek@berkeley.edu}} \and $\qquad$ Venkatesan Guruswami\thanks{VG is with the Simons Institute for the Theory of Computing, and the Departments of EECS and Mathematics, University of California, Berkeley. Contact: \href{mailto:venkatg@berkeley.edu}{venkatg@berkeley.edu}}}
\begin{document}
\maketitle


\begin{abstract}
In the modern era of large-scale computing systems, a crucial use of error correcting codes is to judiciously introduce \emph{redundancy} to ensure \emph{recoverability} from failure. To get the most out of every byte, practitioners and theorists have introduced the framework of \emph{maximal recoverability} (MR) to study optimal error-correcting codes in various architectures. In this survey, we dive into the study of two families of MR codes: MR locally recoverable codes (LRCs) (also known as partial MDS codes) and grid codes (GCs).

For each of these two families of codes, we discuss the primary recoverability guarantees as well as what is known concerning optimal constructions. Along the way, we discuss many surprising connections between MR codes and broader questions in computer science and mathematics. For MR LRCs, the use of skew polynomial codes has unified many previous constructions. For MR GCs, the theory of higher order MDS codes shows that MR GCs can be used to construct optimal list-decodable codes. Furthermore, the optimally recoverable patterns of MR GCs have close ties to long-standing problems on the structural rigidity of graphs.
\end{abstract}


\section{Introduction}

In the design of storage architectures, one often encodes data using suitable error-correcting codes to ensure global data integrity in spite of many local failures. The fundamental goal of the study of error-correcting codes is to understand the tradeoff between \emph{redundancy} and \emph{recoverability}. In this survey, we focus on the study of \emph{linear} error-correcting codes, where the code $C$ is a subspace of the vector space $\F_q^n$, where $\F_q$ is a finite field. If the dimension of $C$ is $k$, we say $C$ is an $[n,k]$ code. For example, the seminal family of Reed--Solomon codes \cite{reed1960polynomial} consider the evaluations of a degree $< k$ polynomial at $n$ evaluation points in $\F_q$, where $q \ge n \ge k$. It is known that Reed--Solomon codes have optimal recoverability among all $[n,k]$ codes in that they attain the \emph{Singleton bound} (i.e., they are maximum distance separable (MDS) codes)~\cite{singleton1964maximum}. More precisely, 
\emph{every} subset of $n-k$ erased symbols can be recovered by interpolation on the remaining $k$ symbols.

This optimality of Reed--Solomon codes is a special case of a more general notion called \emph{maximal recoverability}.  We say that a pattern $E \subseteq [n] := \{1,2,\hdots, n\}$ is \emph{recoverable} (interchangeably \emph{correctable}) if for any $c \in C$, if we erase the coordinates indexed by $E$ from $c$, we can still uniquely identify $c$. For this survey, we introduce the notation $\recov(C)$ to denote the set of recoverable patterns. For any family $\mathcal{F}$ of $[n,k]$-codes, we let $\recov(\mathcal{F})$ denote the \emph{union} of $\recov(C)$ for all $C$ in $\mathcal{F}$.

 Note that recoverability is monotone: if $E \subseteq [n]$ is recoverable for $C$ and $E' \subseteq E$, then $E'$ is recoverable for $C$. Since $C$ is a linear code, a partial converse is true: if $E \subseteq [n]$ is recoverable for $C$, then there always exists $E' \supseteq E$ of size $n-k$ which is recoverable for $C$. As such, $\recov(C)$ (and thus $\recov(\mathcal{F})$) is completely determined by sets $E \in \recov(C)$ of size $n-k$. We now define maximal recoverability. We remark that the term ``maximally recoverable'' was coined by Chen, Huang, and Li~\cite{chen2007maximally,huang2013pyramid} in the context of distributed storage although the notion was 
 inspired by earlier results in the field of \emph{network coding}~\cite{koetter2003algebraic,jaggi2005polynomial}.

\begin{definition}[MR]
An $[n,k]$-code $C$ is \emph{maximally recoverable (MR)} for $\mathcal{F}$ if $\recov{C} = \recov{\mathcal{F}}$. 
\end{definition}

In particular, \emph{any} $[n,k]$-Reed--Solomon code is maximally recoverable for the family $\mathcal{F}$ of all $[n,k]$-codes. In general, given a family $\mathcal{F}$, it is not always clear whether such a MR code $C$ exists! However, for our applications, which are families of codes defined by the topological constraints of a data center, the family $\mathcal{F}$ has sufficient algebraic structure  
that there are infinitely many maximally recoverable codes, among which we seek to identify and construct the codes of minimal field size. In combinatorics, such questions are closely related to the theory of the \emph{representability of matroids}, see Section~\ref{subsec:matroid} and Section~\ref{sec:concl} for more details.
We also note that it is not always clear for a given family $\mathcal{F}$ if there is an efficient description of the recoverable patterns $\recov(\mathcal{F})$.

We now dive into the study of MR codes by carefully studying two families of MR codes: MR \emph{Locally Recoverable Codes} (MR LRC) and MR Grid Codes (MR GC). Although the two families of codes are quite similar in definition, the underlying nature of the codes are quite different. For example, merely describing the recoverable patterns of a MR LRC is a routine exercise, while doing the same for a MR GC is {an} open question in matroid theory.

\subsection{Preliminaries}

To establish some common notation, given a $[n,k]$-code $C \subseteq \F_q^n$, we define a \emph{generator matrix} of $C$ to be a matrix $G \in \F_q^{k \times n}$ such that the $k$ rows of $G$ form a basis for $C$. {T}he generator matrix of $C$ is not unique, but given any two generators $G_1, G_2 \in \F_q^{k \times n}$, there always exists an invertible matrix $M \in \F_q^{k \times k}$ such that $G_2 = M G_1$.

We let $\langle u, v\rangle := \sum_{i=1}^n u_iv_i$ denote the dot product between two vectors $u, v \in \F_q^n$. Given a $[n,k]$-code $C \subseteq \F_q^n$, we define its dual code $C^{\perp} \subseteq \F_q^n$ to be $\{c' \in \F_q^n \mid \forall c \in C, \langle c, c'\rangle = 0\}$. One can show that $C^{\perp}$ is always a $[n,n-k]$-code. We typically use the letter $H \in \F_q^{(n-k) \times n}$ to denote a generator matrix of $C^{\perp}$, which is also the \emph{parity check matrix} of $C$. Note that $c \in C$ if and only if $Hc = 0^{n-k}$.

Given a matrix $M \in \F_q^{k \times n}$ and a set $S \subseteq [n]$, we let $M|_{S}$ denote the $k \times |S|$ matrix consisting of the columns of $M$ indexed by $S$. Given a $[n,k]$-code $C$ and a generator matrix $G$, we let the \emph{puncturing} $C|_{S}$ denote the code generated by $G|_{S}$. We also note that puncturing the parity check matrix gives a simple test for recoverability.

\begin{proposition}\label{prop:recov-rank}
Let $H$ be the parity check matrix of a $[n,k]$-code $C$. We have that $E \subseteq [n]$ is recoverable if and only if the columns of $H$ spanned by $E$ are linearly independent--that is, $\rank H|_{E} = |E|$. In particular, when $|E| =n-k$, $E$ is recoverable iff $\det (H|_{E}) \neq 0$.
\end{proposition}

By convention, we define a \emph{parity check code} to be the $[n,n-1]$ code which has the row vector $1^n$ as a parity check matrix.

\section{Maximally Recoverable Locally Recoverable Codes}

One disadvantage of an MDS code such as a Reed--Solomon code is that it lacks \emph{locality}. That is, to correct even a single erasure in an $[n,k]$-Reed--Solomon code, one needs to read $k$ other symbols. {An alternative family of codes which reduces the number of reads is that of} \emph{locally recoverable codes (LRC)} (e.g.,~\cite{chen2007maximally,huang2012erasure}). The topology of an LRC consists of four parameters $(n,r,a,h)$. Here, $n$ is the number of symbols, $r$ is the size of each local groups ($g := n/r$ total groups), $a$ is the number of parity checks per local group, and $h$ is the number of global parity checks---a fail-safe for more than $a$ erasures in a local group. The dimension of such a code is $k = n-(n/r)a-h$. We further impose that each of the groups are ``locally MDS'' in the sense that any $r-a$ symbols within a group can reconstruct the other $a$ symbols.

A more precise description of an MR LRC, which was also studied earlier under the name partial MDS codes~\cite{BlaumPS13}, can be made in terms of its parity check matrix. For all $i \in [g]$, consider matrices $A_i \in \F_q^{a \times r}$ and $B_i \in \F_q^{h \times r}$. The parity check matrix of our MR LRC is {then} (cf.~\cite{gopi2020maximally})
\begin{equation}
    \label{eq:MR-LRC-parity-check}
H := \left[\begin{array}{c|c|c|c}
A_1 & 0 & \cdots & 0\\\hline
0 & A_2 & \cdots & 0\\\hline
\vdots & \vdots & \ddots & \vdots\\\hline
0 & 0 & \cdots & A_{g}\\\hline
B_1 & B_2 & \cdots & B_g
\end{array}\right].
\end{equation}

\smallskip
Equivalently, from an encoding point of view, an $(n,r,h,a)$-LRC is obtained by adding $h$ global parity checks to $k$ data symbols, partitioning these $k+h$ symbols into local groups of size $r-a$, and then adding `$a$' local parity checks for each local group. 
See Figure~\ref{Fig:LRC} for an illustration. 

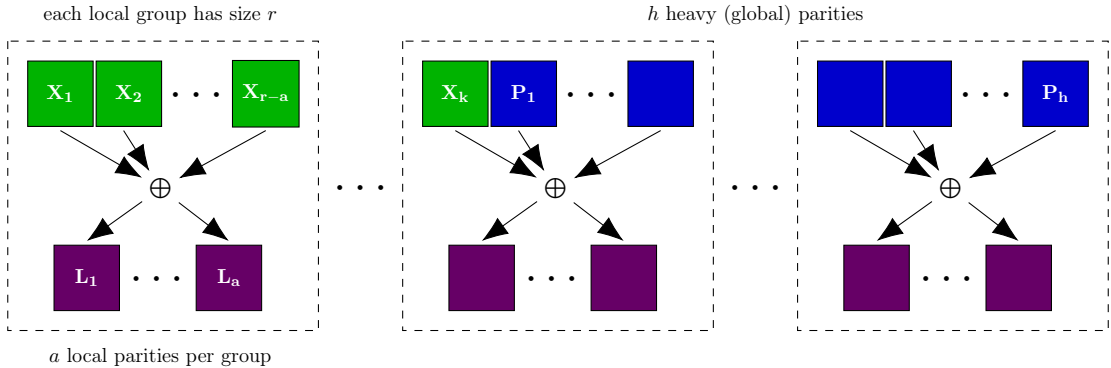
\begin{figure*}[h]

\begin{center}
\begin{tikzpicture}[scale=0.7,square/.style={regular polygon, regular polygon sides=4,minimum size=50,inner sep=0},every node/.style={transform shape}]

  \def\A{0}
  \def\B{7.5}
  \def\C{15}

  \node at (\A+2.9,0.5) {each local group has size $r$};
  \node at (\A+2.9,-6.0) {$a$ local parities per group};
  \node at (\B+6.7,0.5) {$h$ heavy (global) parities};

  \node at (\A+1,-1) [square,draw,fill=green!70!black] (X1) {\color{white}$\mathbf{X_1}$};
  \node at (\A+2.3,-1) [square,draw,fill=green!70!black] (X2) {\color{white}$\mathbf{X_2}$};
  \node at (\A+3.65,-1.03) {\Huge$\cdots$};
  \node at (\A+4.9,-1) [square,draw,fill=green!70!black] (Xra) {\color{white}$\mathbf{X_{r-a}}$};

  \node at (\A+2.9,-2.8) (xor1) {$\bigoplus$};
  \draw [-{Latex[length=10]}] (\A+1,-1.7)--(xor1);
  \draw [-{Latex[length=10]}] (\A+2.2,-1.7)--(xor1);
  \draw [-{Latex[length=10]}] (\A+4.9,-1.7)--(xor1);

  \node at (\A+1.5,-4.5) [square,draw,fill=violet!80!black] (L1) {\color{white}$\mathbf{L_1}$};
  \node at (\A+2.9,-4.53) {\Huge$\cdots$};
  \node at (\A+4.2,-4.5) [square,draw,fill=violet!80!black] (La) {\color{white}$\mathbf{L_a}$};
  \draw [-{Latex[length=10]}] (xor1)--(\A+1.5,-3.8);
  \draw [-{Latex[length=10]}] (xor1)--(\A+4.2,-3.8);

  \draw[dashed] (\A,0) rectangle (\A+5.9,-5.5);

  \node at (\A+6.77,-2.8) {\Huge$\cdots$};

  \node at (\B+1,-1) [square,draw,fill=green!70!black] (Xk) {\color{white}$\mathbf{X_k}$};
  \node at (\B+2.3,-1) [square,draw,fill=blue!80!black] (X2) {\color{white}$\mathbf{P_1}$};
  \node at (\B+3.65,-1.03) {\Huge$\cdots$};
  \node at (\B+4.9,-1) [square,draw,fill=blue!80!black] (Xra) {};

  \node at (\B+2.9,-2.8) (xor1) {$\bigoplus$};
  \draw [-{Latex[length=10]}] (\B+1,-1.7)--(xor1);
  \draw [-{Latex[length=10]}] (\B+2.2,-1.7)--(xor1);
  \draw [-{Latex[length=10]}] (\B+4.9,-1.7)--(xor1);

  \node at (\B+1.5,-4.5) [square,draw,fill=violet!80!black] (L1) {};
  \node at (\B+2.9,-4.53) {\Huge$\cdots$};
  \node at (\B+4.2,-4.5) [square,draw,fill=violet!80!black] (La) {};
  \draw [-{Latex[length=10]}] (xor1)--(\B+1.5,-3.8);
  \draw [-{Latex[length=10]}] (xor1)--(\B+4.2,-3.8);

  \draw[dashed] (\B,0) rectangle (\B+5.9,-5.5);

  \node at (\B+6.77,-2.8) {\Huge$\cdots$};

  \node at (\C+1,-1) [square,draw,fill=blue!80!black] (X1) {};
  \node at (\C+2.3,-1) [square,draw,fill=blue!80!black] (X2) {};
  \node at (\C+3.65,-1.03) {\Huge$\cdots$};
  \node at (\C+4.9,-1) [square,draw,fill=blue!80!black] (Xra) {\color{white}$\mathbf{P_h}$};

  \node at (\C+2.9,-2.8) (xor1) {$\bigoplus$};
  \draw [-{Latex[length=10]}] (\C+1,-1.7)--(xor1);
  \draw [-{Latex[length=10]}] (\C+2.2,-1.7)--(xor1);
  \draw [-{Latex[length=10]}] (\C+4.9,-1.7)--(xor1);

  \node at (\C+1.5,-4.5) [square,draw,fill=violet!80!black] (L1) {};
  \node at (\C+2.9,-4.53) {\Huge$\cdots$};
  \node at (\C+4.2,-4.5) [square,draw,fill=violet!80!black] (La) {};
  \draw [-{Latex[length=10]}] (xor1)--(\C+1.5,-3.8);
  \draw [-{Latex[length=10]}] (xor1)--(\C+4.2,-3.8);

  \draw[dashed] (\C,0) rectangle (\C+5.9,-5.5);

\end{tikzpicture}
\end{center}
\caption{This diagram visualizes the encoding map of a locally recoverable code. First, $k$ data symbols are expanded to $k+h$ symbols using $h$ global parity checks. Then, these symbols are broken up into $\frac{k+h}{r-a}$ local groups, each of which is expanded with $a$ more parity checks. Illustration adapted from Figure 1 of \cite{gopi2020maximally}.}\label{Fig:LRC}
\end{figure*}

Subject to these restrictions, what does it mean for an LRC to be maximally recoverable? That is, which erasure patterns are recoverable? It turns out there is a simple description (e.g., \cite{gopalan2014explicit}).
\begin{theorem}\label{thm:MR-LRC}
In an $(n,r,a,h)$-MR LRC, a pattern $E \subseteq [n]$ is recoverable, if and only if there exists $E' \subseteq E$ of size at least $|E|-h$ such that each local group intersects with $E'$ in at most $a$ symbols.
\end{theorem}

\begin{proof} {First we prove that if $E$ is recoverable, then $E'$ exists.} By Proposition~\ref{prop:recov-rank} it suffices to determine which subsets of the columns of $H$ are linearly independent for some choice of $A_1, \hdots, A_g$ and $B_1, \hdots, B_g$ {with $g = n/r$}. Assume that $E \subseteq [n]$ {is a recoverable erasure pattern so} that $\rank H|_{E} = {|}E{|}$, {and assume for sake of contradition that no claimed $E'$ exists.} {As such, there is} ${F} \subseteq E$ {which spans}{$\ell \le g$} local groups of $E$ {but} $|{F}| > a\ell + h$. Note that $H|_{{F}}$ has at most $a\ell + h$ nonzero rows. Therefore, $\rank H|_{{F}} < |{F}|$, a contradiction of the fact that $E$ and thus ${F}$ is correctable. {Therefore, the claimed $E'$ indeed exists.} 

{Conversely, we prove the existence of $E'$ implies $E$ is recoverable. Assume} that $E = E_0 \cup E_1 \cup \cdots \cup E_g \subseteq [n]$ has $a$ symbols $E_i$ from each local group $i \in [g]$ plus $h$ additional symbols $E_0$. It suffices to exhibit one specific choice of $A_1, \hdots, A_g, B_1, \hdots, B_g$ for which $H|_{E}$ has rank $|E|$. One way to do this is as follows:

\begin{itemize}
\item For all $i \in [g]$, set $A_i$ to be an MDS matrix such that $A_i|_{E_i}$ to be a copy of the identity matrix. 
\item Set $\left.\begin{bmatrix}B_1 & B_2 & \cdots & B_g\end{bmatrix}\right|_{E_0}$ to be a copy of $I_h$, with the remainder of the row block equal to zero.
\end{itemize}

With this choice of $H$, the submatrix $H|_{E}$ is a permutation matrix and thus has rank $|E|$. Thus, $E$ is recoverable in a $(n,r,a,h)$-MR LRC.
\end{proof}

With the recoverable patterns fully characterized, the primary challenge in the study of MR LRCs as follows: given parameters $(n,r,a,h)$ what is the size of the smallest field size $q$ for which an $(n,r,a,h)$-MR LRC exists over $\F_q$?
We begin with a random coding result over large fields, which is in fact prototypical of the way one establishes the \emph{existence} of MR codes in the first place.
\begin{theorem}\label{thm:mr-lrc-random}
In the parity check matrix $H$ from \eqref{eq:MR-LRC-parity-check}, let the matrices $A_i$ all be equal to the same parity check matrix of some $[r,r-a]$ MDS code {over $\F_q$}, say a Reed-Solomon code.
Let the matrices $B_i$, $i \in [g]$, be filled in with independent uniformly random entries from $\F_q$. Then provided $q \gg \binom{r}{a}^{n/r} n^h$, the resulting code is $(n,r,a,h)$-MR LRC with high probability.
\end{theorem}
The proof idea is as follows. {In \eqref{eq:MR-LRC-parity-check}, replace each $B_i$ with distinct indeterminate variables $X_{t,j}$, $t \in [h]$, $j \in [n]$. For every maximal correctable erasure pattern $E$ meeting the criteria of Theorem~\ref{thm:MR-LRC}, we seek to prove $\det H|_{E}$ is a nonzero polynomial of degree at most $h$.} One can do column operations to zero out all columns within each $A_i$ except the first $a$ columns erased within group $i$. The resulting determinant will be a nonzero multiple of a determinant of an $h \times h$ matrix whose entries are some linear forms in the $X_{t,j}$'s. Thus, if we assign the indeterminates random values from $\F_q$, the determinant will vanish with probability at most $\frac{h}{q}$ by the Schwarz-Zippel lemma. There are in total at most $\binom{r}{a}^{n/r}\binom{n - ar}{h}$ correctable patterns and thus determinants to worry about, so the chance any of them vanishes for random $\F_q$-values is at most $h \binom{r}{a}^{n/r} \binom{n}{h}/q$. For $q \gg \binom{r}{a}^{n/r} n^h$, this tends to $0$.


The result above has two significant drawbacks---it is non-constructive and only works over very large fields. There has been a huge amount of work giving varied, often incomparable,  MR LRC constructions, tailored to different parameter regimes (whether $r$ is close to $n$ or $\ll n$, how small $h$ is, how $h$ and $r$ compare, etc.) 
In terms of settings most relevant to distributed storage practice, one should think of the number of local groups $g=n/r$ as a constant and $n$ as growing. Typical values of $g$ used in practice are $g=2,3,4$. The number of global parities $h$ is also a small constant and the number of local parities $a$ is usually $1$ or $2$. For example, an early version of Microsoft's Azure storage used $(n=14,r=7,h=2,a=1)$-MR LRCs with $g=2$ local groups~\cite{huang2012erasure}.  These choices are dictated by the need to maximize storage efficiency while balancing reliability and fast reconstruction.  Contrast this with the parameters of interest from a theoretical point of view, where $r$ is sublinear in $n$ or even a constant in order to get good locality, and one might consider large $h$ to correct many erasures.

Returning to constructions of MR LRCs, let us articulate the main challenge. If our goal is to only recover from $a$ erasures in each group, taking each $A_i$ to be a Vandermonde matrix achieves this.
If the goal is to recover from $h$ arbitrarily distributed global erasures, that's also easy by simply taking the heavy parities (the $B_i$'s) to be an $h \times n$ Vandermonde matrix. The challenge is to be able to handle any combination of these local and global erasures. 

One elegant construction of MR LRCS, due to \cite{gopalan2014explicit}, is to take the heavy parities to be a \emph{Moore matrix}. Compared to a Vandermonde matrix where we take successive powers of the first row, in a Moore matrix we successively apply the Frobenius automorphism to each preceding row. The following is an $h \times n$ Moore matrix, where we require the $\alpha_i \in \F_q$ in the first row to be \emph{linearly independent} over a subfield $\F_p$ of $\F_q$ (which is stronger than the distinctness requirement in the Vandermonde case. This requires $q \ge p^n$ and to keep the field size small, we in fact take $q = p^n$.
\begin{equation*}
\left[\begin{array}{cccc}
\alpha_1 & \alpha_2 & \dots  & \alpha_n \\ 

\alpha_1^p & \alpha_2^p & \dots  & \alpha_n^p \\ 

\vdots & \vdots & \ddots & \vdots\\

\alpha_1^{p^{h-1}} & \alpha_2^{p^{h-1}} & \dots  & \alpha_n^{p^{h-1}} 
\end{array}\right]
\end{equation*}
It is a well-known algebraic fact that any $h$ columns of the above Moore matrix are linearly independent over $\F_q$ (the larger field). This aspect is similar to the guarantee offered by Vandermonde matrices. The additional power offered by the Moore matrix in the context of MR LRCs is the following: 
if we modify any column of the above matrix using column operations (with $\F_p$ coefficients), the resulting matrix is also a Moore matrix (this follows because $\alpha \mapsto \alpha^p$ is a $\F_p$-linear map on $\F_q$). One can then show, using a reasoning similar to the one used in the initial part of the random coding argument, that if the $A_i$'s are Vandermonde matrices over $\F_p$ (which we can achieve with $p = O(r)$, and even $p=2$ when $a=1$), the overall construction is an MR LRC. This achieves a field size $q \le O(r)^n$ (and at most $2^n$ in the case where $a=1$).
 
A closer inspection of the proof reveals that we don't really need all the $\alpha_i$'s to be linearly independent, and it suffices if any $(a +1)h$ of them are
linearly independent over $\F_p$. This can be achieved with $q \le O(n)^{(a+1) h}$ using suitable Reed-Solomon or BCH codes. In fact, for $a=1$, we can take $p=2$ and use binary BCH codes and achieve $q \le O(n)^{\lfloor (a+1)h/2\rfloor}$; see \cite{gopalan2014explicit} for details.

For small values of $h$ (a setting relevant in practice), one can get better field sizes, e.g., $O(r)$ for $h =0,1$~\cite{BlaumPS13}, $O(n)$ for $h=2$ and $O(n^3)$ for $h=3$~\cite{gopi2020maximally}. For a small number of local groups $g = n/r$ (also relevant in practice), a field size of $O(n)^{(g-1)(a+h/g)}$ can be achieved~\cite{hu2016new,dhar2023construction}.

See Table 1 of \cite{dhar2023construction} and Table 1 of \cite{cai2021construction} for pointers to several constructions of MR LRCs tailored to different parameter settings.
For most ranges of parameters, the best current constructions in terms of field size are due to \cite{martinez2019universal,gopi2022improved,cai2021construction}. In particular, a field size of 
\begin{equation}
    \label{eq:skew-poly-field-size}
O(\max\{n/r, r\}))^{\min\{h,r-a\}}
\end{equation}
can be achieved~\cite{gopi2022improved}, based on the theory of skew polynomials. 
Skew polynomials are a non-commutative analog of polynomials studied by Ore back in 1933. They are defined by the ring $\mathbb{K}[x;\sigma]$ of polynomials in $x$ with coefficients over $\mathbb{K}$ and right multiplication by a scalar coefficient defined as $x \cdot a = \sigma(a) x$ for a field homomorphism $\sigma : \mathbb{K} \to \mathbb{K}$. The construction of \cite{gopi2022improved} works with the specific case where $\mathbb{K} = \F_{p^m}$ is an extension field and $\sigma(a) = a^p$ is the Frobenius automorphism. One can partition $\mathbb{K}^*$ into $p-1$ ``conjugacy classes" such that any non-zero degree $d$ skew polynomial has roots in at most $d$ distinct conjugacy classes, and at most $d$ $\F_p$-linearly independent roots in any single conjugacy class. In the associated MR LRC construction, the local groups corresponds to these conjugacy classes. This approach yields 
matrices that blend together Vandermonde and Moore matrices in a manner conducive to MR LRCs over smaller fields; see \cite{cai2021construction, gopi2022improved} for details.

Despite significant efforts, the best constructions for most regimes still require super-linear, and typically $n^{\Omega(h)}$, field sizes. Is this necessary? In \cite{gopi2020maximally}, it was shown that the field size $q$ of an $(n,r,a,h)$-MR LRC must obey the lower bound $q  \ge \Omega_{h,a}(n \cdot r^{\min\{a, h-2\}})$ (for $h \le n/r$).  This shows that for small $h \le a+2$ and $r = n^{\Omega(1)}$, one needs a field size of $n^{\Omega(h)}$, giving some justification for the large field sizes in the constructions. In fact, when $h$ is a fixed constant with $h \le a+2$ and $r,r-a = \Theta(\sqrt{n})$, the field size \eqref{eq:skew-poly-field-size} achieved by the skew polynomials based construction becomes $O(n)^{h/2}$, which asymptotically matches the above lower bound! Besides the case of $h=2$ where the optimal field size is $\Theta(n)$, this is the only case where we know the optimal field size.

\section{Maximally Recoverable Grid and Tensor Codes}\label{sec:MRG}


A natural extension of MR LRCs is to have multiple sets of local repair groups. Such an extension was formulated in the context of codes with grid-like topologies by Gopalan et al. \cite{Gopalan2016}. More precisely, given parameters $(m,n,a,b,h)$, we can construct a grid code with $mn$ symbols, which we identify as the entries of an $m \times n$ matrix. For each of the $n$ columns we place $a$ parity checks on the $m$ symbols within that column. Likewise, for each of the $m$ rows, we place $b$ parity checks on the $n$ symbols within that row. {The parity checks are the same between columns and between rows.} Finally, we add $h$ additional parity checks, so the dimension of our code is $k = (m-a)(n-b) - h$. Such a topology has been used by Meta~\cite{MLRH14} {with $(m,n,a,b,h) = (3,14,1,4,0)$}. Note that if $a = 0$, we have the topology of an $(mn, n, b, h)$ LRC. {Going forward,} we always assume that $a, b \ge 1$.

\subsection{Characterizations of Correctable Patterns}

Unlike LRCs, recoverable patterns of MR grid codes (GCs) are generally poorly understood; however, major progress was made by Holzbauer, Puchinger, Yaakobi, and Wachter-Zeh~\cite{holzbaur2021correctable}.
\begin{theorem}[\cite{holzbaur2021correctable}]\label{thm:no-h}
A pattern $E \subseteq [m] \times [n]$ is correctable for an $(m,n,a,b,h)$-MR GC if and only if there is $E' \subseteq E$ with $|E'| \ge |E| - h$ that is correctable for an $(m,n,a,b,0)$-MR GC.
\end{theorem}

Thus, {to obtain} a description of the correctable patterns for MR GCs, it suffices to consider the case $h = 0$. These are known as $(m,n,a,b)$-\emph{MR tensor codes (TCs)}, and have been a frequent topic of interest in the literature. As the name suggests, a $[mn, (m-a)(n-b)]$-code $C$ which is an $(m,n,a,b)$-MR tensor code can be written\footnote{Here, we define the tensor product of two codes to be the code whose generator matrix is the Kroenecker product of the constituent generator matrices.} as $C_{\col} \otimes C_{\row}$, where $C_{\col}$ is a $[m,m-a]$-code and $C_{\row}$ is a $[n,n-b]$-code. In other words, $C_{\col}$ enforces the $a$ parity checks per column, and $C_{\row}$ enforces the $b$ parity checks per row. See Figure~\ref{fig:MRG}{.}

\begin{figure}
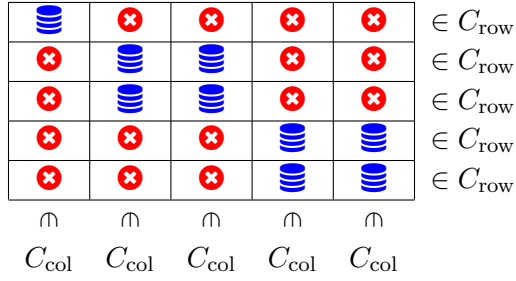

\def\arraystretch{1.1}
  \begin{center}
     \begin{tabular}{|c|c|c|c|c|c}\cline{1-5}
      \gd & \bd & \bd & \bd & \bd & $\in C_{\row}$\\\cline{1-5}
      \bd & \gd & \gd & \bd & \bd & $\in C_{\row}$\\\cline{1-5}
      \bd & \gd & \gd & \bd & \bd & $\in C_{\row}$\\\cline{1-5}
      \bd & \bd & \bd & \gd & \gd & $\in C_{\row}$\\\cline{1-5}
      \bd & \bd & \bd & \gd & \gd & $\in C_{\row}$\\\cline{1-5}
            \multicolumn{1}{c}{$\forkv$}
      & \multicolumn{1}{c}{$\forkv$}
        & \multicolumn{1}{c}{$\forkv$}
          & \multicolumn{1}{c}{$\forkv$}
            & \multicolumn{1}{c}{$\forkv$} &\\
      \multicolumn{1}{c}{$C_{\col}$}
      & \multicolumn{1}{c}{$C_{\col}$}
        & \multicolumn{1}{c}{$C_{\col}$}
          & \multicolumn{1}{c}{$C_{\col}$}
            & \multicolumn{1}{c}{$C_{\col}$} &
    \end{tabular}
  \end{center}
\caption{An illustration of a pattern which is not correctable in a $(5,5,2,2)$-MR tensor code as discovered by Holzbauer et al.~\cite{holzbaur2021correctable}. Each \bd\ represents an erased symbol.}\label{fig:MRG}
\end{figure}

\def\arraystretch{1.25}

\subsubsection{Case Study: \texorpdfstring{$a=b=1$}{a=b=1}}\label{subsec:a=b=1}

Before we describe some general results for MR tensor codes, we start with the simplest special case: where $a$ and $b$ are both $1$. In other words, every row and column of our $m \times n$ grid has a single parity check. Let's now investigate which patterns $E \subseteq [m] \times [n]$ are correctable. Fix a codeword $c \in \F_q^{m \times n}$ and assume the symbols $c|_{E}$ have been erased. As a key change in perspective, think of $E$ as the edges of a bipartite graph whose vertices are $V := [m] \sqcup [n]$, the disjoint union of $[m]$ and $[n]$. If a vertex of $(V,E)$ has its degree equal to $1$, then some row or column of the grid has a single error $e \in E$. In that case, we can use the parity check of the row or column to recover $e \in E$. As such, the recovery problem for $E$ is now reduced to the recovery problem for $E' := E \setminus \{e\}$.

By recursively applying this procedure, it suffices to study the recoverability of graphs $(V,E)$ such that each vertex has degree zero or degree at least $2$. In that case, either our graph has zero edges or has a cycle. The former case is trivially correctable, as no symbols were erased. However, if the graph has a cycle, then the pattern is not correctable. 

To see why, it suffices to exhibit a nonzero codeword $c$ whose support (i.e., nonzero coordinates) lies entirely in $E$. We can build such a codeword by alternately assigning $\pm 1$ along the edges of one such cycle and setting all other coordinates of the codeword to be $0$. Thus, any pattern $E$ that contains a cycle is not correctable. Unrolling our recursive argument, we have {proved} the following theorem. 

\begin{theorem}[e.g., \cite{Gopalan2016}]\label{thm:MR-grid}
A pattern $E \subseteq [m] \times [n]$ is correctable for an $(m,n,1,1)$-MR tensor code if and only if $E$ lacks a cycle.
\end{theorem}

Furthermore, our argument shows that the tensor product of any two parity check codes (over any field) is a $(m,n,1,1)$-MR tensor code. We now build on this case study to consider $b > 1$.

\subsubsection{The \texorpdfstring{$a=1$}{a=1} Case}

The more general setting of $a = 1$ and $b \ge 1$ was studied in the original paper on MR GCs by Gopalan et al. \cite{Gopalan2016}. {To lead up to this result, we} first understand the {limitations} of the $a=b=1$ analysis{.}

Recall that we can {view} a given pattern $E \subseteq [m] \times [n]$ as a bipartite graph. {When} $a=b=1$, we could WLOG assume that no vertex has degree equal to $1$. {For $a=1$ and general $b$, we can modify this trick as follows.} {If} any column of $E$ has $1$ symbol, we can immediately correct it. Likewise, if any of the $m$ rows of $E$ has $b$ or fewer symbols, we can immediately correct those. Thus, in the graph interpretation of $E$, if we think of $[m]$ as the ``left'' side and $[n]$ as the ``right'' side, then we can without loss of generality assume that the degree of each vertex on the right side is at least $2$, and the degree of each vertex on the left side is at least $b+1$.

With this reduction, one may hope like in the $a=b=1$ case that such patterns $E$ are either empty or non-recoverable, but this is not the case. For example, for $(m,n,a,b) = (3,{6},1,2)$, the following pattern is recoverable even though every column has at least $2 = a+1$ erasures and every row has at least $3 = b+1$ erasures.

\begin{center}
     \begin{tabular}{|c|c|c|c|c|c|}\cline{1-6}
      \bd & \bd & \bd & \gd & \gd & \gd\\\cline{1-6}
      \bd & \gd & \gd & \bd & \bd & \gd\\\cline{1-6}
      \gd & \bd & \gd & \bd & \gd & \bd\\\cline{1-6}
      \gd & \gd & \bd & \gd & \bd & \bd\\\cline{1-6}
    \end{tabular}
\end{center}

{The justification that such patterns are correctable requires much more subtle understanding of what makes $(C_{\col} \otimes C_{\row})|_{\bar{E}}$ have full dimension.} A key contribution by Gopalan et al. \cite{Gopalan2016} was describing these correctable patterns using a property which is known as \emph{regularity}.\footnote{As we shall discuss in Section~\ref{subsec:matroid}, Kalai, Nevo, and Novik~\cite{kalai2015Bipartite} independently discovered this property, which they called a \emph{Laman condition}, in the context of {bipartite rigidity}.}

\begin{definition}\label{def:reg}
A pattern $E \subseteq [m] \times [n]$ is $(a,b)$-regular if for all $S \subseteq [m]$ and $T \subseteq [n]$ with $|S| \ge a$ and $|T| \ge b$, {we have that} $|E \cap (S \times T)| \le a|T| + b|S| - ab$.
\end{definition}

It is straightforward to show for any choice of $a,b \ge 1$ that if $E \subseteq [m] \times [n]$ is recoverable in an $(m,n,a,b)$-tensor code, then $E$ is $(a,b)$-regular. The proof follows by an information-theoretic version of the pigeonhole principle. Assume that $E$ is recoverable but not regular. Thus, there is $S \subseteq [m]$ and $T \subseteq [n]$ with $|S| \ge a$ and $|T| \ge b$ but $|E \cap (S \times T)| > a|T| + b|S| - ab$. Subtracting $|S \times T|$ from both sides, we get that
\begin{align}
  |(S \times T) \setminus E| < (|S|-a)(|T|-b).\label{eq:ST-pigeon}
\end{align}
However, our code only has $a$ parity checks per column and $b$ parity checks per row, so any $|S| \times |T|$-sized box must have at least $(|S|-a)(|T|-b)$ symbols of information.
Thus, (\ref{eq:ST-pigeon}) implies that some information was destroyed, contradicting that $E$ is recoverable.

The other direction---that regularity implies recoverability---is much more difficult to establish. Proving this direction for $a=1$ and general $b$ is one of the main results in \cite{Gopalan2016}.
\begin{theorem}[\cite{Gopalan2016}]\label{thm:reg}
A pattern $E \subseteq [m] \times [n]$ is correctable for an $(m,n,1,b)$-MR tensor code if and only $E$ is $(1,b)$-regular.
\end{theorem}

The key idea behind the proof of Theorem~\ref{thm:reg} is a careful application of Hall's marriage theorem. We remark that an alternative proof of Theorem~\ref{thm:reg} was established by Brakensiek, Gopi, and Makam using the theory of higher order MDS codes~\cite{bgm2021mds,bgm2022}, which are explained in more detail in Section~\ref{subsec:homds}. They also give multiple polynomial-time algorithms for checking whether a given pattern $E \subseteq [m] \times [n]$ is $(1,b)$-regular and thus recoverable in an $(m,n,1,b)$-MR tensor code.

\subsubsection{The \texorpdfstring{$a \ge 2$}{a >= 2} Case}

One might hope that Theorem~\ref{thm:reg} extends for general $(a,b)$. Such a statement was conjectured by Gopalan et al.~(and even partially verified~\cite{shivakrishna2018Maximally}) but such hopes turned out to be false. Holzbauer, Puchinger, Yaakobi, and Wachter-Zeh~\cite{holzbaur2021correctable} showed that for the parameters $(m,n,a,b) = (5,5,2,2)$, there exists a pattern $E \subseteq [5] \times [5]$ which is $(2,2)$-regular but not recoverable. See Figure~\ref{fig:MRG} for a description of the pattern.

As such, even for $a=b=2$, the problem of characterizing the correctable patterns is vastly more challenging. A purely combinatorial description for $a=b=2$ is given by Bernstein~\cite{bernstein2017completion}, although its applicability to MR tensor codes was only recently observed by Brakensiek, Dhar, Gao, Gopi, and Larson~\cite{brakensiek2024rigidity}. See Section~\ref{subsec:matroid} for more details on Bernstein's original motivation.

\begin{theorem}[\cite{bernstein2017completion} as stated in \cite{brakensiek2024rigidity}]\label{thm:bernstein}
A pattern $E \subseteq [m] \times [n]$ is correctable in an $(m,n,2,2)$-MR tensor code if and only if it has a two-coloring $\sigma : E \to \{1,2\}$ of the edges with the following properties.
\begin{itemize}
\item No cycle of $E$ is monochromatic (in the sense that every edge is assigned the same color).
\item No cycle of $E$ is alternating (in the sense that every vertex has an edge colored $1$ and an edge colored $2$).
\end{itemize}
\end{theorem}


Compared to the combinatorial methods of Theorem~\ref{thm:reg}, Theorem~\ref{thm:bernstein} is proved using machinery of \emph{matroid theory} and \emph{tropical geometry}. As remarked by \cite{brakensiek2024rigidity}, finding a polynomial-time description of Theorem~\ref{thm:bernstein} is open, although a conditional algorithm was recently provided by Brakensiek, Chen, Dhar, and Zhang~\cite{brakensiek2025random}; see Section~\ref{subsec:rigid} for further discussion. For other recoverability results, such as descriptions of recoverable patterns when $m$ is bounded in addition to $a$ and $b$, see \cite{brakensiek2024rigidity} and the references therein.

\begin{proof}[Sketch of Proof of Theorem~\ref{thm:bernstein}]
For simplicity of this sketch, we assume that $E$ has size $2m+2n-4$, so that by Proposition~\ref{prop:recov-rank} the condition that $E$ is correctable is equivalent to $\det(H|_{E}) \neq 0$, where $H \in \F^{(2m+2n-4) \times mn}$ is the parity check matrix of an $(m, n, 2, 2)$-MR tensor code.

\paragraph{``If'' direction.} We first prove that if $E$ has a two coloring $\sigma : E \to \{1,2\}$ with the prescribed properties (which we henceforth call a ``good'' coloring), then $\det(H|_{E}) \neq 0$. For the special case in which $\F$ is the complex field (i.e., a field of characteristic zero), this direction was proved by Bernstein~\cite{bernstein2017completion} using tropical geometry. However, we present a more elementary argument due to Brakensiek, Dhar, Gao, Gopi, and Larson~\cite{brakensiek2024rigidity} which applies to every field characteristic.

The first step is to express the parity check matrix in terms of the parity check matrices of the underlying $[m, m-2]$ and $[n, n-2]$ codes. More precisely, assume the parity check matrix of $C_{\col}$ and $C_{row}$ are
\[
\begin{pmatrix}
x_{1,1} & \cdots & x_{m, 1}\\
x_{1,2} & \cdots & x_{m, 2}
\end{pmatrix}\text{ and }\begin{pmatrix}
y_{1,1} & \cdots & y_{n, 1}\\
y_{1,2} & \cdots & y_{n, 2}
\end{pmatrix},
\]
respectively, where we think of these entries as variables over a polynomial ring. Then, for each $(i, j) \in [m] \times [n]$, the $(i,j)$th column of $H$ has four\footnote{An astute reader will notice that this description of $H$ involves $2m+2n$ rows rather than $2m+2n-4$. Four of these rows need to be deleted to get a true parity check matrix, but we omit such technical details from this sketch.} nonzero entries:
\[
  H_{2i-1, (i,j)} = y_{j,1}, H_{2i, (i,j)} = y_{j,2}, H_{2m+2j-1, (i,j)} = x_{i,1}, H_{2m+2j, (i,j)} = x_{i,2}.
\]
As such, when we compute $\det(H|_{E})$ in terms of these symbolic variables, each monomial in the determinant expansion requires picking one of these four variables $\{x_{i,1}, x_{i,2}, y_{j,1}, y_{j,2}\}$ for each $(i,j) \in E$, subject to the constraint that each row has exactly one monomial selected. Using our good two-coloring $\sigma : E \to \{1,2\}$, we select our monomials using the following procedure.

\begin{itemize}
\item First, let $T_1, \hdots, T_\ell$ monochromatic components of $E$ with respect to $\sigma$. Since $E$ has no monochromatic cycle, each $T_k$ is a tree.
\item Assign an arbitrary root $r_k \in V(T_k)$ for each tree and pick an orientation of the edges of $T_k$ such that all edges are oriented away from $r_k$ (like in a breadth-first search).
\item For each $(i,j) \in E$, we select the monomial corresponding to $(i,j)$ as follows.
\begin{itemize}
  \item If the edge is oriented $i \to j$, pick $x_{i,\sigma(i,j)}$.
  \item If the edge is oriented $j \to i$, pick $y_{j,\sigma(i,j)}$. 
\end{itemize}
\end{itemize}
The product of these monomials, which we call $M$, appears in the determinant expansion of $\det(H|_{E})$.  Conversely, any monomial of $\det(H|_{E})$ can be viewed an oriented two-coloring of $E$. However, with a careful analysis, one can show $M$ is special in the sense that $M$ appears exactly once in the determinant expansion of $\det(H|_{E})$ and thus cannot cancel out, so $\det(H|_{E}) \neq 0$. 

To prove this, assume for contradiction that $M$ appears a second time in the expansion. That means there is a second oriented two-coloring of $E$ yielding this monomial. By carefully tracking how the monomials correspond to the structure of the graph, one can deduce this is only possible if the second two-coloring reverses some alternating cycles of $E$. However, our hypothesis about $E$ states that it lacks alternating cycles, thus $M$ is indeed unique.

\paragraph{``Only if'' direction.} A first question one might ask is if the combinatorial proof we just gave can prove the converse. The answer is no, although we show that our special monomial does not cancel out if $E$ has a special coloring, we do not know (at least currently) a direct proof that the monomials corresponding to `bad' colorings of $E$ always cancel out.

As such, we present Bernstein's proof using tropical geometry~\cite{bernstein2017completion}. We again assume that $E \subseteq [m] \times [n]$ has size $2m+2n-4$. Like in the proof of the ``if'' direction, we think of $\det(H|_{E})$ as an arithmetic circuit using only symbolic variables $\{x_{i,c}, y_{j,c}\}$  and the arithmetic operators $\{+, -, \times\}$. Thus, the circuit being zero in characteristic zero implies it is zero in all characteristics. In other words, we may without loss of generality think of $\det(H|_{E})$ as a symbolic polynomial over $\mathbb C$.

In tropical geometry, we replace this arithmetic circuit which with its \emph{tropicalization}, where (in general) the $+/-$ operators become the $\max$ operator, and the $\times$ operator becomes the $+$ operator.  With such a transformation, a complicated algebraic function turns into a system of (partial) hyperplanes stitched together. We recommend the textbook of Maclagan and Sturmfels for an introduction to the field~\cite{maclagan2015introduction}. A crucial property of this tropicalization is that it preserves the underlying matroid structure of $H$~\cite{yu2017algebraic}, that is for every $E \subseteq [m] \times [n]$, whether the columns of $H|_{E}$ are linearly independent is captured in the structure of the tropicalization.

Using a number of tools in tropical geometry, one can recast the structure of MR tensor codes into the structure of a two-dimensional Grassmannian, more precisely the set of $2 \times 2$ determinants of the  $2 \times (m+n)$ matrix
\[
\begin{pmatrix}
x_{1,1} & \cdots & x_{m, 1} & y_{1,1} & \cdots & y_{n, 1}\\
x_{1,2} & \cdots & x_{m, 2} & y_{1,2} & \cdots & y_{n, 2}
\end{pmatrix}.
\]
Crucially for us, the tropicalization of the two-dimensional Grassmannian is very well understood~\cite{maclagan2015introduction}. In particular, the primary structure of the Grassmannian corresponds to the structure of \emph{phylogenetic trees} with $m+n$ leaves. That is, trees with $m+n$ and $m+n-2$ internal nodes each of degree $3$, so $2m+2n-3$ edges total. Using this property of the tropical Grassmannian, Bernstein shows that a pattern $E \subseteq [m] \times [n]$ is correctable if and only if there exists a phylogenetic tree $T$ with $m+n$ leaves (indexed by $[m] \sqcup [n]$) with the following property:
\begin{itemize}
\item Let $Q$ be a $|E| \times |T|$ matrix such that for every edge $(i,j) \in E$ and every edge $e \in T$, we have that $Q_{(i,j), e} = 1$ if $e$ lies on the unique path from $i$ to $j$ in $T$ and $Q_{(i,j), e} = 0$ otherwise. Then, $E$ is correctable if and only if the rank\footnote{Since $Q$ has $|T|=2m+2n-3$ columns, one might think this implies there are correctable patterns with $2m+2n-3$ leaves. However, a rank of $2m+2n-3$ is only possible of $E$ has at least one edge which connects two vertices of $[m]$ or two vertices of $[n]$. This does not make sense in the context of tensor codes, but this does make sense in the context of the more general ``skew-symmetric matrix completion'' problem which Bernstein studies. See \cite{bernstein2017completion,brakensiek2024rigidity} for more details.} of $Q$ is $|E|$.
\end{itemize}
Although the above property is combinatorial, it is not quite enough to deduce the prescribed coloring of $E$ when $\det(H|_{E}) \neq 0$. Bernstein performs another trick in which he shows the above fact does not merely hold for some tree $T$, but rather $T$ can be a more structured \emph{caterpillar tree} where the $m+n-2$ internal nodes all lie on the same path. Walking along these $m+n-2$ internal nodes, we induce an ordering on the $m+n$ leaves (with an inconsequential ambiguity on how to order the first two and last two leaves). Effectively, this caterpillar tree $T$ induces a permutation $\pi$ on the vertex set $[m] \sqcup [n]$. 
We now color $E$ as follows: if $(i,j)$ in $E$ has $\pi(i) < \pi(j)$, set $\sigma(i,j) = 1$. Otherwise, set $\sigma(i,j) = 2$. Bernstein shows that since $\rank Q = |E|$ and $T$ is a caterpillar tree, this two coloring has no monochromatic cycle and no alternating cycle, as desired.
\end{proof}

\subsection{Field Size Bounds and Explicit Constructions}\label{subsec:field-size}

Now that we have described the correctable patterns of MR tensor and GCs, we next turn to a more practical question: over what field sizes $q$ do these $(m,n,a,b,h)$ MR GCs exist? And, can we construct such codes explicitly? For simplicity of discussion, we assume that $a$ and $b$ are constant.

\subsubsection{MR Tensor Codes} First we look at the case in which $h=0$, i.e., the setting of MR tensor codes. Recall by our discussion in Section~\ref{subsec:a=b=1}, that the case in which $a=b=1$ can realized by the tensor product of two parity checks. As such, $(m,n,1,1)$-MR tensor codes can be explicitly constructed over \emph{every} field, with the smallest being $\F_2$.

However, this clean picture is unique to $a=b=1$, as Kong, Ma, and Ge~\cite{kong2021new} showed that every $(m,n,1,2)$-MR tensor code with $m \ge 4$ requires $q = \Omega(n^2)$. Further, Brakensiek, Dhar, and Gopi~\cite{bdg2024size} showed that for $m \ge 3$, any $(m,n,1,b)$-MR tensor code requires field size at least $\Omega(n^{b-1})$.

More results are known in terms of upper bound for general $(m,n,a,b)$-MR tensor codes. Such bounds are typically proved using techniques similar to the proof of Theorem~\ref{thm:mr-lrc-random}. More precisely, one shows that $C_{\col} \otimes C_{\row}$ being a $(m,n,a,b)$-MR tensor code is equivalent to $f(m,n,a,b)$ nonzero polynomials of degree at most $d$ having a nonzero evaluation at certain entries of the generator (or parity check) matrices of $C_{\col}$ and $C_{\row}$ , where each polynomial typically corresponds to the correctability of some pattern $E \subseteq [m] \times [n]$. By the Schwarz-Zippel lemma, we can then set $q \approx d \cdot f(m,n,a,b)$. Using such an argument, Kong, Ma, and Ge~\cite{kong2021new} proved that when $n \gg m$, $(m,n,1,b)$-MR tensor codes exist over field size $\approx n^{b(m-1)}$. Brakensiek, Gopi, and Makam~\cite{bgm2021mds} generalized this by showing for $n \gg m$, $(m,n,a,b)$-MR tensor codes exist over field sizes $\approx n^{b(m-a)}$. We also remark that Athi, Chigullapally, Krishnan, and Lalitha~\cite{athi2023Structurea} have improved upper bounds in the regime for which {$n-b$ is much larger than $m$, where $q \approx 2^{m^2(n-b)}$ suffices.}

Only a few explicit constructions are known. In the case in which $m$ and $n$ are both growing, the best explicit constructions have size doubly exponential in $m$ and $n$ (e.g., \cite{shangguan2020combinatorial,roth2021higher,bdg2024size}). However, when $m=3$ and $b \le 5$, some $n^{O(1)}$-sized explicit constructions are known. For a more detailed overview of known field size bound for MR tensor codes and the closely related higher order MDS codes (see Section~\ref{subsec:homds}), please refer to \cite{bdg2024size}.

\subsubsection{MR Grid Codes}

Much less is known for field size bounds in the setting $h \ge 1$. Holzbauer, Puchinger, Yaakobi, and Wachter-Zeh~\cite{holzbaur2021correctable} show that if an $(m,n,a,b)$-MR tensor code exists over field size $q$, then the code can be transformed into a $(m,n,a,b,h)$-MR GC over field size $q^{(m-a)(n-b)}$ for \emph{any} $h \ge 1$.

We remark this exponential increase in field size is actually necessary in the case $a=b=1$. More precisely, Kane, Lovett, and Rao~\cite{kane2019independence} prove the following result.

\begin{theorem}\label{thm:KLR}
If $C \in \F_{q}^{n\times n}$ is a $(n,n,1,1,1)$-MR GC then $q \ge \Omega(2^{n/2})$.
\end{theorem}

They further show that one can pick $q = O(8^n)$ for infinitely many $n$ (and the construction is explicit).  A follow-up paper by Coregliano and Jeronimo~\cite{coregliano2022tighter} improved the lower bound to $q \ge \Omega(1.97^{n})$. {See also a very recent study by Brakensiek, Dhar, and Gopi~\cite{brakensiek2025improved} of the setting where $m$ is much smaller than $n$.}

\begin{proof}[Sketch of Proof of Theorem~\ref{thm:KLR}]
An essential idea by Kane, Lovett, and Rao is that constructing an $(n,n,1,1,1)$-MR grid code is equivalent to finding a coloring of the \emph{Birkhoff polytope graph} $\mathcal B_n$ whose vertices are the elements of the symmetric group $S_n$, i.e., the family of all $n!$ permutations of $n$. Two permutations $\pi, \tau \in S_n$ are connected by an edge if $\sigma := \pi^{-1} \circ \tau$ is a cycle permutation, i.e., there are $i_1, \hdots, i_m \in [n]$ such that $\sigma$ maps $i_1 \to i_2 \to \cdots \to i_m \to i_1$ and all other indices are fixed by $\sigma$.

To see why this is the case, recall that a $(n,n,1,1,1)$-MR grid code has exactly one ``global'' parity check. This parity check can be interpreted as a labeling $\gamma : [n] \times [n] \to \F_q$, where $\gamma(i,j) \in \F_q$ is the $(i,j)$th entry of the global parity check vector. Using Theorem~\ref{thm:MR-grid}, one can show that $\gamma$ is indeed the parity check of a $(n,n,1,1,1)$-MR grid code if and only if for every (simple) cycle graph $E \subseteq [n] \times [n]$, the alternating sum of the terms $\gamma(i,j)$ for edges $(i,j) \in E$ is nonzero. Henceforth, we call such a labeling $\gamma$ \emph{valid}.

 Kane, Lovett, and Rao prove in Claim~1.5 of \cite{kane2019independence} that given a valid $\gamma$, $\mathcal B_n$ has a valid $q$-coloring via the map $\pi \mapsto \sum_{i=1}^n \gamma(i, \pi(i))$. In particular, by the pigeonhole principle, the existence of a valid $\gamma$ implies the existence of an independent set of $\mathcal B_n$ of size $n! / q$.

Thus, to lower bound $q$, it suffices to upper bound the size of the largest independent set of $\mathcal B_n$. Let $A \subseteq S_n$ be a maximum-sized independent set. The high-level strategy used by Kane, Lovett, and Rao is to consider the following pseudorandomness test:
\begin{itemize}
\item Given even $m \le n$ and distinct $i_1, \hdots, i_m \in [n]$ and distinct $j_1, \hdots, j_m \in [n]$, how likely is it for a random $\pi \in A$ to map $\pi(i_a) = j_a$ for all $a \in [m]$?
\end{itemize}
If there is a choice of $m$ and indices such that the answer is larger than $\frac{2^{m/2} \cdot m!}{n!}$, then the set $A$ fails the pseudorandomness test, and the authors find an independent set $A' \subseteq S_{n-m}$ which is denser than it should be. Otherwise, $A$ looks pseudorandom, and the authors bound the size of $A$ using the crucial tool of \emph{representation theory}. One can define a group ring $\mathbb R[S_n]$ consisting of \emph{class functions} of the form $\varphi = \sum_{\pi \in S_n} a_\pi \pi$, where each $a_{\pi} \in \mathbb R$. To test whether $A \subseteq S_n$ is an independent set, Kane, Lovett and Rao construct two class functions
\begin{align*}
\varphi_A &:= \frac{1}{|S_n|} \frac{1}{|A|^2} \sum_{\substack{\sigma \in S_n\\\pi, \pi' \in A}} \sigma \pi (\pi')^{-1} \sigma^{-1}\\
\psi_n &:= \frac{1}{|C_n|} \sum_{\tau \in C_n} \tau,
\end{align*}
where $C_n$ is the set of cycles of $S_n$ of length $n$. Crucially, one can show that $A$ being an independent set implies\footnote{This condition only checks for edges corresponding to full-length cycles, so it is not equivalent to being an independent set.} the inner product $\langle \varphi_A, \psi_n\rangle = 0$. To understand this inner product, one can decompose the $\varphi_A$ and $\psi_n$ into \emph{Specht modules} (which can be thought of as a non-Abelian analogue of a Fourier basis), with the coefficients of these Spech modules being the \emph{characters}. By standard results in the representation theory of the symmetric group, these characters correspond to the \emph{Young tableaux} of size $n$. Furthermore, because $\psi_n$ is only a sum of full-length cycles, the only characters that need to be considered are the much smaller set of \emph{hook tableaux}.  Because of the pseudorandomness property assumed for $A$, the hook tableaux characters of $\varphi_A$ can be computed to suitable precision which enables showing that whenever $|A| > n!/2^{n/2}$, one must have $\langle \varphi_A, \psi_n\rangle > 0$, a contradiction to the fact that $A$ is an independent set. Therefore every independent set in $\mathcal{B}_n$ must have size at most $n!/2^{n/2}$.
\end{proof}

The follow-up work by Coregliano and Jeronimo~\cite{coregliano2022tighter} builds on the techniques of Kane--Lovett--Rao by introducing some new ideas. Instead of just considering $\psi_n$ (i.e., cycles of length $n$), Coregliano and Jeronimo consider others class functions corresponding to cycles of length $n - \ell$ for some parameter $\ell \ge 0$. This makes the representation theory considerably more complicated as a richer family of characters need to be considered. However, together with the pseudorandomness property for $A$, one can also impose many more constraints on the characters of $\varphi_A$. In general, such constraints are difficult to analyze by hand, so  Coregliano and Jeronimo automate the computations using linear programming. The authors strongly suggest that the limit of their techniques should yield a bound of $n! / 2^{n-o(1)}$, but no finite linear program can certify such a bound, causing them to stop short at approximately $n! / 1.97^n$.

\subsection{Higher Order MDS Codes}\label{subsec:homds}

As we previous discussed, due to the results of Gopalan et al.~\cite{gopalan2014explicit}, we have a full description of the correctable patterns of $(m,n,1,b)$-MR tensor codes. This leads to the next natural question: which codes $C_{\col}$ and $C_{\row}$ have the property that $C_{\col} \otimes C_{\row}$ is a $(m,n,1,b)$-MR tensor code?

As a first step, note that $C_{\col}$ is a $[m,m-1]$-code. As such, we can assume essentially WLOG that $C_{\col}$ is the parity check code: it has $1^m$ as its parity check matrix. In other words, whether $C_{\col} \otimes C_{\row}$ is a $(m,n,1,b)$-MR tensor code depends entirely on the structure of $C_{\row}$.

To get some intuition for what properties $C_{\row}$ should have, let's look at the special case of $m=2$. For any $c \in C_{\col} \otimes C_{\row}$, the $C_{\col}$ parity check enforces that $c_{1,i} = -c_{2,i}$ for all $i \in [n]$. Therefore, up to a scaling of the second row, $c$ is a repetition encoding of some $c' \in C_{\row}$. Now, consider an arbitrary erasure pattern $E \subseteq [2] \times [n]$ which is $(1,b)$-regular. Note that for all $i \in [n]$, a symbol $c'_i$ is only lost if both $(1,i),(2,i) \in E$. Let $T \subseteq [n]$ be the maximum-sized set for which $[2] \times T \subseteq [n]$, (i.e., the symbols of $T$ are ``lost''). By Definition~\ref{def:reg}, we have for $S = [2]$ and $T$, we have that $2|T| = |E \cap (S \times T)| \le |T| +2b-b$. Thus, $|T| \le b$. In other words, the property of $C_{\col} \otimes C_{\row}$ being a $(2,n,1,b)$-MR tensor code is equivalent to the property that $C_{\row}$ can recover from any $b$ symbols being erased. Since $C_{\row}$ is a $[n,n-b]$-code, this is equivalent to $C_{\row}$ attaining the Singleton bound. In other words, we proved the following.

\begin{proposition}\label{prop:MDS-2}
$C_{\col} \otimes C_{\row}$ is a $(2,n,1,b)$-MR tensor code if and only if $C_{\row}$ is MDS.
\end{proposition}

For example, we can set $C_{\row}$ to be any Reed-Solomon code with the appropriate length and dimension. It is not hard to prove by similar logic that for any $m \ge 3$, if $C_{\col} \otimes C_{\row}$ is a $(m,n,1,b)$-MR tensor code then $C_{\row}$ is MDS. However, the converse is not true. For example, the following pattern is correctable by any $(3,6,1,3)$-MR tensor code, but there exist Reed-Solomon codes $C_{\row}$ for which $C_{\col} \otimes C_{\row}$ cannot correct this pattern.
\begin{center}
     \begin{tabular}{|c|c|c|c|c|c|}\cline{1-6}
      \bd & \bd & \bd & \bd & \gd & \gd\\\cline{1-6}
      \bd & \bd & \gd & \gd & \bd & \bd\\\cline{1-6}
      \gd & \gd & \bd & \bd & \bd & \bd\\\cline{1-6}
    \end{tabular}
\end{center}
Therefore, for $m \ge 3$, the property of $C_{\col} \otimes C_{\row}$ being a $(m,n,1,b)$-MR tensor code is a strictly stronger condition than $C_{\row}$ being MDS. This observation led to Brakensiek, Gopi, and Makam (BGM)~\cite{bgm2021mds} defining a notion of \emph{higher order MDS} (higher order MDS). 

\begin{definition}[\cite{bgm2021mds}]
We say that a $[n,n-b]$-code $C$ is $\MDS(m)$ (i.e., higher order MDS of order $m$) if $C_{\col} \otimes C$ is a $(m,n,1,b)$-MR tensor code, where $C_{\col}$ is the $[m,m-1]$ parity check code.
\end{definition}

By Proposition~\ref{prop:MDS-2}, we have that $\MDS(2)$ is equivalent to ``ordinary'' MDS. However, $\MDS(3)$ is a strictly stronger notion. The paper~\cite{bgm2021mds} establishes a number of properties of higher order MDS codes. 
For example, it is shown that for any code $C$ which is $\MDS(3)$, its dual code is also $\MDS(3)$, generalizing the analogous property of MDS codes. However, such duality fails for $\MDS(m)$ for $m \ge 4$, although some weaker results can be proved.

\subsubsection{Connections to List-decoding}

Soon after \cite{bgm2021mds} appeared, Roth~\cite{roth2021higher} independently formulated a notion of higher order MDS codes in a seemingly rather different context: \emph{average-radius list decoding}. Concretely, given a $[n,k]$-code $C$ and parameters $\tau, L \in \mathbb N$, we say that $C$ is $(\tau, L)$-average radius list decodable if for all distinct $c_{1}, \hdots, c_{L+1} \in C$ and $y \in \F_q^n$, we have that $\wt(c_1 - y) + \cdots + \wt(c_{L+1} - y) > \tau (L+1),$ where $\wt$ denoted Hamming weight.
A natural question is given fixed parameters $n,k,L$, what is the optimal value of $\tau$? Note that for $L=1$, average-radius list decoding simply demands that the code has minimum distance at least $2 \tau$, and thus the optimal $\tau$ equals $(n-k)/2$ by the Singleton bound. In general\footnote{For the purposes of this survey, we ignore issues concerning extreme parameters{--}see Roth~\cite{roth2021higher} for more details.}, the optimal choice of $\tau$ is $\tfrac{L(n-k)}{L+1}$. This leads to the following definition.

\begin{definition}[\cite{roth2021higher}]
An $[n,k]$-code $C$ is\footnote{Brakensiek, Gopi, and Makam~\cite{bgm2022} used the notation ``$\LDMDS(L)$'' to refer to such codes.} $L$-MDS if $C$ is $(\tfrac{L(n-k)}{L+1}, L)$-average radius list decodable.
\end{definition}

Roth used the notion of $L$-MDS to abstract a conjecture of Shangguan and Tamo~\cite{shangguan2020combinatorial} concerning the optimal list-decodability of Reed-Solomon codes. In Roth's language, the Shangguan-Tamo conjecture asserts that for all $n,k,L$ there exists a $[n,k]$-Reed-Solomon code $C$ over some field $\F_q$ which is $L$-MDS.

Serendipitously, Roth proves a number of notions concerning $L$-MDS codes which were similar to results proved in \cite{bgm2021mds}. For example, $1$-MDS are equivalent to MDS, while $2$-MDS is a strictly stronger notion. Likewise, $2$-MDS codes are dual to $2$-MDS codes, but this fails for $3$-MDS codes. {These facts} led Brakensiek, Gopi, and Makam~\cite{bgm2022} to find a deeper reason for the coincidence: {duality}!

\begin{theorem}[\cite{bgm2022}]\label{thm:ld-mds}
For all $L \in \mathbb N$, a code $C$ is $\MDS(L+1)$ iff $C^{\perp}$ is $\ell$-MDS for all $\ell \le L$.
\end{theorem}

{P}roving the Shangguan-Tamo conjecture is {then} equivalent\footnote{The dual of a Reed-Solomon code is only a \emph{generalized} Reed-Solomon code, but if a GRS code is {$L$-MDS} then there exists a corresponding RS code which is $L$-MDS.} to showing for all $n,b,m$ there exists a $[n,n-b]$-Reed-Solomon code $C_{\row}$ for which $C_{\col} \otimes C_{\row}$ is a $(m,n,1,b)$-MR tensor code! In the same paper~\cite{bgm2022}, the authors manage to prove such a result by connecting higher order MDS codes {with the} \emph{GM-MDS theorem} {in coding theory}.  

\subsubsection{Connections to the GM-MDS Theorem}

Recall that a given $[n,k]$-code $C \subseteq \F_q^{n}$ can have many generator matrices $G \in \F_q^{k \times n}$. One combinatorial way to distinguish these different generator matrices is to look at their \emph{zero patterns}. We define the zero pattern of a matrix $G$, which {we} denote by $\zp(G)$ is a list of $k$ subsets $(S_1, \hdots, S_k)$ of $[n]$, where $S_i$ for each $i \in [k]$ is the set of indices $j \in [n]$ for which $G_{i,j} = 0$. We {let} $\ZP(C)$ denote the set of all zero patterns $\zp(G)$ where $G$ is a generator matrix of $C$. The set $\ZP(C)$ reveals much information about the structure of $C$. For example if $C$ has a zero pattern $\zp(G) = (S_1, \hdots, S_k)$ where some $|S_i| \ge k$, then the $i$th row of $G$ is a codeword of $C$ with Hamming weight at most $n-k$. Therefore, $C$ cannot be MDS.

A natural question asked by Dau, Song, and Yuen~\cite{dau2015simple} is given a MDS $C$, which zero patterns are attainable? They prove the following.

\begin{theorem}[\cite{dau2015simple}]\label{thm:gm}
For any $S_1, \hdots, S_k \subseteq [n]$, we have that $(S_1, \hdots, S_k) \in \ZP(C)$ for some MDS code $C$ if and only if for all nonempty $I \subseteq [k]$, we have that $\left|\bigcap_{i \in I} S_i\right| \le k - |I|$.
\end{theorem}

Note this result leads to a question similar to maximal recoverability, namely which codes $C$ attain all (or at least many) of the zero patterns identified by Theorem~\ref{thm:gm}?  \cite{bgm2022} give a precise answer to this question. We say that $(S_1, \hdots, S_k)$ is an \emph{order-$m$ zero pattern} if it satisfies the inequalities of Theorem~\ref{thm:gm} and there are most $m$ distinct non-empty sets among $S_1, \hdots, S_k$.

\begin{theorem}[\cite{bgm2022}]\label{thm:mds-gm}
For all $m \ge 2$, a code $C$ is $\MDS(m)$ if and only if every order-$m$ zero pattern is attained by $C$.
\end{theorem}

Combining Theorem~\ref{thm:ld-mds} and Theorem~\ref{thm:mds-gm}, one can resolve the Shangguan-Tamo conjecture by showing that there exist Reed-Solomon codes attaining every generic zero pattern. This question, however, was already answered independently by Lovett~\cite{lovett2018gmmds} and Yildiz and Hassibi~\cite{yildiz2019gmmds} in the form of the GM-MDS theorem.

\begin{theorem}[GM-MDS Theorem, \cite{lovett2018gmmds,yildiz2019gmmds}]\label{thm:gm-mds}
For every zero pattern identified by Theorem~\ref{thm:gm}, there is a Reed-Solomon code attaining that zero pattern.
\end{theorem}

As a corollary (via an application of the Schwarz-Zippel lemma similar to that of Theorem~\ref{thm:mr-lrc-random}), 
there exists a Reed-Solomon code over an exponential size field that satisfies \emph{all} possible zero patterns, and thus is an $\MDS(m)$ code for all $m \ge 2$. The dual of this code {resolves} the Shangguan-Tamo conjecture.

\subsection{Further Applications}\label{subsec:matroid}

So far, we have discussed in detail various properties of MR tensor codes. We now discuss broader applications of {MR tensor codes.}

\subsubsection{Further Applications to List Decoding}

Although the paper~\cite{bgm2022} resolved Shangguan-Tamo's conjecture, it did not fully resolve the question of the minimum field size for which Reed-Solomon codes can achieve \emph{(average radius) list-decoding capacity}. More precisely, instead of requiring our code $C$ to be precisely $(\tfrac{L(n-k)}{L+1}, L)$-average radius list decodable, what if we relax the condition to being $((1-\eps)(n-k), O(1/\eps))$-average radius list decodable, for some $\eps \in (0,1)$? {Alrabiah, Guo, Guruswami, Li, and Zhang~\cite{alrabiah2025random} show} that random Reed-Solomon codes achieve list-decoding capacity at field size $q = n + k \cdot 2^{O(1/\eps^2)}$.

We note that this relaxation parameter $\eps$ is necessary, as the MR tensor code lower bound of Brakensiek, Dhar, and Gopi~\cite{bdg2024size} showed precisely optimal list-decodability requires exponential field size. In fact, Alrabiah, Guruswami, and Li~\cite{alrabiah2023ag} showed that being $((1-\eps)(n-k), O(1/\eps))$-average radius list decodable requires a field size at least $\exp(\Omega(1/\eps))$, even if the code is non-linear! This nearly matches an upper bound of $\exp(O(1/\eps^2))$ they establish for random linear codes.

We note that these list-decoding results inspired Brakensiek, Dhar, Gopi, and Zhang~\cite{brakensiek2025ag} to define a notion of \emph{relaxed} higher order MDS codes in the sense that not every possible pattern needs to be correctable. This framework {was used to show} that other families of codes can achieve list-decoding capacity, including random algebraic geometry codes~\cite{brakensiek2025ag}.

\subsubsection{Applications to Variable Packet-error Coding}

Very recently, higher order MDS codes were applied by Kong, Wang, Roth, and Tamo~\cite{kong2025new} in the context of \emph{variable packet-error coding} (VPEC) arising from motivations in the field of network error correction (NEC). Here, the objective of VPEC is to transmit data across many packets, where the redundancy within each packet dynamically adjusts to take advantage of fluctuations in the error rate of transmission.  One of their constructions involves \emph{interleaving} multiple higher order MDS codes {together.}

\subsubsection{Connections to Structural Rigidity Theory}\label{subsec:rigid}

So far, all the connections we have described for MR tensor codes have been within the realm of error-correcting codes. However, recent work has shown that the study of correctable patterns in MR tensor codes is equivalent to well-studied questions in \emph{matroid theory}, more precisely in \emph{structural rigidity theory}.

Informally, an undirected graph $G = (V,E)$ is \emph{$d$-rigid}\footnote{As far as the authors are aware, graph rigidity is unrelated to the computer science problem of \emph{matrix rigidity}.} if there is an embedding\footnote{There's a precise requirement that the embedding must be sufficiently ``generic.'' For instance, one cannot map all the points to the same line.} of the vertices $V$ into $d$-dimensional space $\R^d$ such that every infinitesimal motion of the vertices changes the length of some edge (except for rotations and translations). For example, three points forming a path is \emph{not} $2$-rigid because one can flex the two edges around the central vertex. A complete classification of 2-rigid graphs was established by Pollaczek-Geiringer~\cite{PollaczekGeiringer} and Laman~\cite{Laman}. A recent article by Brakensiek, Eur, Larson, and Li~\cite{BELL25} derived a novel proof of this classification using the GM-MDS theorem via connections to the algebraic geometry concept of \emph{cross-ratio degrees} (e.g., \cite{JK15,CGS20,S22}).

The study of $3$-rigid graphs was initiated by James Clerk Maxwell in 1864~\cite{maxwell64,maxwell}. It is a long-standing open question to give a deterministic polynomial-time characterization of $3$-rigid graphs. See the survey of Cruickshank, Jackson, Jord{\'a}n, and Tanigawa~\cite{cruickshank2025rigidity} for further background on rigidity problems.

In a related vein, a recent result by Brakensiek, Dhar, Gao, Gopi, and Larson~\cite{brakensiek2024rigidity} builds a precise connection between {MR} and \emph{bipartite rigidity}, coined by Kalai, Nevo, and Novik~\cite{kalai2015Bipartite} {where} one embeds the two sides of a bipartite graph into orthogonal spaces.

\begin{theorem}[\cite{brakensiek2024rigidity}, informal]\label{thm:rigid}
A pattern $E \subseteq [m] \times [n]$ of size $(m-a)(n-b)$ in a $(m,n,a,b)$-MR tensor code {over $\F$} is correctable\footnote{The geometric definition of rigidity only makes sense in characteristic zero fields like $\R$, but there are alternative definitions for finite fields for which the equivalence still holds.} if and only if there is a bipartite rigid embedding of $E$ into $({\F}^a, {\F}^b)$.
\end{theorem}

The bridge built by Theorem~\ref{thm:rigid} connected various results for bipartite rigidity and MR tensor codes. For example, Bernstein's Theorem (Theorem~\ref{thm:bernstein}) was proved in the context of bipartite rigidity, but could be ported\footnote{The reduction is not immediate since \cite{bernstein2017completion} only considers $\R$, but \cite{brakensiek2024rigidity} patches the gap using combinatorial methods.} to MR tensor codes via Theorem~\ref{thm:rigid}. Furthermore, an equivalent version of Theorem~\ref{thm:reg} is proved in \cite{kalai2015Bipartite}. See \cite{brakensiek2024rigidity} for a more thorough discussion of the rigidity literature as well as connections to other problems such as the \emph{matrix completion problem}.

A systematic understanding of which graphs are bipartite rigid is a topic of active research. A recent paper by Jackson and Tanigawa~\cite{jackson2024maximal} gave a conditional $\mathsf{coNP}$ characterization of bipartite rigid graphs. In other words, assuming a combinatorial conjecture, if a graph is not bipartite rigid, there is a short certificate (informally a ``covering'' of the bipartite graph) asserting this fact. A subsequent work by Brakensiek, Chen, Dhar, and Zhang~\cite{brakensiek2025random} gave a deterministic polynomial time algorithm for checking bipartite rigidity assuming the truth of a 45-year-old open question due to Mason~\cite{mason1981glueing}. This latter work builds on \cite{brakensiek2024rigidity} by constructing an explicit family of error-correcting codes known as \emph{folded Reed-Solomon codes} \cite{guruswami2008explicit,guruswami2016explicit} which when tensored together appears to simulate MR codes.

\subsubsection{Connections to Quantum Codes}

In quantum computing, \emph{quantum error-correcting codes} provide vital hope to ensure that computations are done accurately even when the underlying qubits are not perfectly reliable. In recent years, a powerful method for constructing quantum codes (specifically, those with low-weight stabilizers) has emerged via the so-called \emph{product expansion} property of classical codes, with the constructions themselves using tensor products of two or more such codes. In particular, this has led to breakthrough constructions of asymptotically good quantum LDPC codes as well as classical locally testable codes~\cite{dinur2022locally,panteleev2022asymptotically}.
In order to study product expansion and its applications to quantum error-correction more systematically, 
Kalachev and Pantelev in a series of works~\cite{kalachev2022two,panteleev2024maximally,kalachev2025maximally} have initiated the study of \emph{maximally extendable (ME)} codes, drawing inspiration from the study of MR codes.

We now give an informal description of the ME property. Consider an $(m, n, a, b)$ tensor code $C_1 \otimes C_2$ as well as a pattern $E \subseteq [m] \times [n]$ and a partial message $c \in \F^E$. We seek to understand when $c$ can be extended to a full codeword $\bar{c} \in C_1 \otimes C_2$. We say that a subset $\ell \subseteq [m] \times [n]$ is a \emph{line} if is of the form $\{(i, j_0) : i \in [m]\}$ for some $j_0 \in [n]$ or of the form $\{(i_0, j) : j \in [n]\}$ for some $i_0 \in [m]$. If $\ell \subseteq E$, then for $c \in \F^m$ to have an extension in $C_1 \otimes C_2$, it must be the case that the restriction $c|_{\ell}$ be contained in either $C_1$ or $C_2$ (depending on the direction of the line). In this case, we say that $c$ passes all line checks. We say that $E$ is extendable if for any $c \in \F^E$ which passes all line checks induced by $E$, we can extend $c$ to a (not necessarily unique) codeword $c \in C_1 \otimes C_2$. Like in the study of MR, we seek to understand which codes $C_1$ and $C_2$ have the property that $C_1 \otimes C_2$ has the maximal number of extendable patterns (i.e., $C_1 \otimes C_2$ is ME). Understanding how our understanding of MR codes could impact the study of ME codes is an important open question moving forward.


\section{Conclusion and Open Questions}\label{sec:concl}

In this survey, we scratched the surface of the beautiful notion of maximal recoverability in coding theory. In the case of MR LRCs, the use of rich algebraic techniques such as skew polynomial codes highlights the nontriviality of optimal constructions. Likewise, in the case of MR GCs, even understanding the combinatorial underpinning leads to diverse connections such as list-decoding, matroid theory, representation theory, and more.

More generally, maximal recoverability is closely tied to \emph{matroid realizability}, where one seeks to find an explicitly specified (typically linear) matroid isomorphic to a given combinatorial matroid. Although studying such questions for general matroids is potentially hopeless, we hope that these bridges built between the maximal recoverability of real-world codes and questions in matroid theory will introduce new synergies between the respective domains. We conclude this survey with a few exciting directions for future exploration.

\paragraph{Improved field-size Bounds for MR LRCs.}

For MR LRCs, when $r=O(1)$ (i.e., constant-sized local groups), we do not know any super-linear in $n$ lower bound on the field size, even though the constructions use much larger fields. Can one narrow this vast gap in our understanding via new lower bound techniques or constructions?

\paragraph{Characterization of correctable patterns in MR GCs.}

As discussed at length, outside of some limited parameter settings, we lack an efficiently-computable characterization of which patterns are correctable in an $(m,n,a,b,h)$ MR grid code. A substantial step forward by Holzbauer, Puchinger, Yaakobi, and Wachter-Zeh~\cite{holzbaur2021correctable} shows that it suffices to characterize the correctable patterns of $(m,n,a,b)$ MR tensor codes. More precisely, a pattern $E \subseteq [m] \times [n]$ is correctable for an $(m,n,a,b,h)$ tensor code if and only if at most $h$ elements of $E$ can be removed to create a pattern correctable for an $(m,n,a,b)$ MR tensor code. Thus, we focus on the tensor code setting.

The (presumably) simplest case for which an efficient characterization is lacking is when $a=b=2$. The breakthrough due to Bernstein~\cite{bernstein2017completion} (Theorem~\ref{thm:bernstein}) gives a novel combinatorial condition for testing correctability, but it is an open problem to make this condition efficient. One potential step toward this goal is a recent \emph{conditional} polynomial time algorithm for testing correctability due to Brakensiek, Chen, Dhar, and Zhang~\cite{brakensiek2025random}. However, one must first resolve an open question due to Mason~\cite{mason1981glueing} on the theory of \emph{abstract tensor matroids}; see the survey of Cruickshank, Jackson, Jord{\'a}n, and Tanigawa~\cite{cruickshank2025rigidity} (particularly Conjecture 5.31) for further discussion. We leave resolving such gaps in the literature as the central open direction in this space.

\paragraph{Improved field size bounds and constructions for higher order MDS codes.}

Returning to the theory of higher order MDS codes, what is the optimal size of a $[n,k]$-code which is $\MDS(m)$ when $k$ and $n$ are related by a multiplicative constant? Approximately, the best lower and upper bounds are approximately $\exp(\Omega(n))$~\cite{bdg2024size} and $\exp(O(mn))$~\cite{bgm2022}, respectively. This factor of $m$ is quite significant, and appears (at least superficially) related to the currently best known lower and upper bounds of $\exp(\Omega(1/\eps))$~\cite{alrabiah2023ag} and $\exp(O(1/\eps^2))$~\cite{alrabiah2025random}, respectively, for the optimal list size for a linear code which is $\eps$-close to list-decoding capacity.

Another significant problem is explicitly constructing higher order MDS codes, for which the best known constructions have field size doubly-exponential in $n$ (when $k$ is of size $\Theta(n)$).

\paragraph{Improved field size bounds and constructions for MR GCs.}

As a final open question, recall we discussed in depth the characterization of correctable patterns of MR grid codes with parameters $(m,n,a,b,h) = (m,n,1,1,1)$. The best known construction with field size $O(8^n)$ is not too different from the best known field size lower bound of $\Omega(1.97^n)$. However, for $(m,n,a,b,h) = (m,n,a,b,h)$ much less is known. The same lower bound of $\Omega(1.97^n)$ applies~\cite{brakensiek2025improved}, but the best known construction has field size $n^{O(n)}$~\cite{gopalan2014explicit,Gopalan2016}. Resolving this gap is an intriguing open question. See the recent work of Brakensiek, Dhar, and Gopi~\cite{brakensiek2025improved} for other open questions in this parameter regime.

\paragraph{Discovering further applications of the MR notion.}

As a final open-ended direction, we encourage the exploration of further connections between maximal recoverability and other problems in mathematics and computer science. From direct technical applications to inspirations, MR has found use in fields as disparate as list decoding~\cite{bgm2022}, structural rigidity~\cite{brakensiek2024rigidity}, and quantum computing~\cite{kalachev2025maximally}. We are eager to see what new discoveries the information theory community might unearth via further exploration of the MR notion in the years to come.

\subsection*{Acknowledgments}

JB and VG are supported in part by a Simons Investigator award and NSF grant CCF-2211972.
JB is also supported by NSF grant DMS-2503280. JB and VG thank Yeyuan Chen, Manik Dhar, Jiyang Gao, Parikshit Gopalan, Sivakanth Gopi, Matt Larson, Visu Makam,  Chaoping Xing, Sergey Yekhanin, and Zihan Zhang for many valuable conversations about MR codes over the years{.} 
{JB and VG further thank anonymous reviewers of IEEE BITS for their feedback.}

\let\oldthebibliography\thebibliography
\let\endoldthebibliography\endthebibliography
\renewenvironment{thebibliography}[1]{
  \begin{oldthebibliography}{#1}
    \setlength{\itemsep}{0.3em}
    \setlength{\parskip}{0em}
}
{
  \end{oldthebibliography}
}

\bibliographystyle{alphaurl}
\bibliography{BITS-ref-clean.bib}



\end{document}